\documentclass[11pt]{article}
\usepackage[T1]{fontenc}
\usepackage[utf8]{inputenc}
\usepackage{lmodern}
\usepackage{amsmath,amssymb,amsthm}
\usepackage[margin=1in]{geometry}
\usepackage{microtype}
\usepackage{hyperref}
\hypersetup{colorlinks=true,linkcolor=blue,citecolor=blue,urlcolor=blue}

\newtheorem{theorem}{Theorem}[section]
\newtheorem{lemma}[theorem]{Lemma}

\newtheorem{proposition}[theorem]{Proposition}
\newtheorem{remark}[theorem]{Remark}

\newcommand{\cl}{\mathcal C}

\title{An Arboricity-Sensitive Algorithm for the $K_r-e$-Free Graph Sandwich Problem}
\author{Min Chih Lin\\
\small Instituto de C\'alculo and Departamento de Computaci\'on,\\
\small Universidad de Buenos Aires, Argentina; CONICET, Argentina\\
\small \texttt{oscarlin@dc.uba.ar}
\and
Nat\'an Vekselman\\
\small Departamento de Computaci\'on, Universidad de Buenos Aires, Argentina\\
\small \texttt{natanvek11@gmail.com}}
\date{}

\begin{document}
\maketitle
\begin{abstract}
For a fixed integer $r\geq4$, the $K_r-e$-free graph sandwich problem asks whether, given graphs $G_1\subseteq G_2$ on the same vertex set, there is an induced-$K_r-e$-free graph $H$ between them. We give a deterministic algorithm taking $O(n+\alpha(G_2)^{r-3}m_2)$ time and space, where $m_2=|E(G_2)|$ and $\alpha(G_2)$ is the arboricity of $G_2$. In particular, the diamond-free case takes $O(n+\alpha(G_2)m_2)$ time. This improves the direct $O(n^r m_2)$ implementation of the previously known forced-edge closure. Our implementation maintains components of common neighborhoods indexed by $(r-3)$-cliques. A filtered frontier supports their merges within the clique-listing bound, while completion events avoid repeatedly searching for affected cliques. On feasible instances the output is contained in every feasible sandwich, independently of processing order. Applying the closure to $(G,K_n)$ gives an $O(n^{r-1})$-time bound for partitioned and nonpartitioned probe $K_r-e$-free recognition, improving the $O(n^{r+2})$ bound obtained from the direct sandwich closure. We also describe a direct static recognizer based on the same local characterization.
\end{abstract}
\noindent\textbf{Keywords:} graph sandwich; forbidden induced subgraph; arboricity; diamond-free graph; probe graph; clique enumeration.

\section{Introduction}
Given graphs $G_1=(V,E_1)$ and $G_2=(V,E_2)$ with $E_1\subseteq E_2$, a graph sandwich for a graph class $\mathcal G$ is a graph $H=(V,E_H)\in\mathcal G$ satisfying $E_1\subseteq E_H\subseteq E_2$ \cite{golumbic1995}. Write $n=|V|$, $m_2=|E_2|$, and $\alpha=\alpha(G_2)$, where the arboricity $\alpha(G)$ is the minimum number of forests whose edge sets partition $E(G)$. We consider $\mathcal G$ to be the class excluding an \emph{induced} $K_r-e$, for fixed $r\geq4$. For $r=4$ this is the diamond-free graph sandwich problem. Dantas et al.\ gave a forced-edge closure algorithm with a direct $O(n^r m_2)$ implementation for fixed $r$ \cite{dantas2011}. Subsequent studies address the complexity of other forbidden-subgraph sandwich problems and related classes \cite{dantas2015,alvarado2019,dantas2019}. Our contribution is a sparse implementation of the $K_r-e$-free closure, with running time governed by clique listing rather than a search over vertex tuples.

\begin{theorem}[Main result]\label{thm:main}
For every fixed $r\geq4$, $K_r-e$-free graph sandwich can be decided, and a solution produced when one exists, in deterministic $O(n+\alpha(G_2)^{r-3}m_2)$ time and space. The output $H_{\min}$ is contained in every feasible sandwich. In particular, for $r=4$ the bound is $O(n+\alpha(G_2)m_2)$.
\end{theorem}

The least-solution property also follows from the earlier forced-edge closure \cite{dantas2011}; we establish that the sparse implementation retains it. For $r=4$, our sandwich bound matches the known $O(n+\alpha(G)m)$ bound for \emph{static} diamond-free recognition \cite{lin2012}, without implying optimality. An application to probe graph recognition is stated in Theorem~\ref{thm:probe}. Couto et al.~\cite{couto2018} also studied the nonpartitioned probe problem; we discuss their stated running time in Section~\ref{sec:probe}. The direct static recognizer in Section~\ref{sec:static} illustrates the local characterization without dynamic edge states. In particular, a specialized $O(n+nm)$ algorithm is available for nonpartitioned probe diamond-free recognition \cite{grippoLin2026}.

An initial membership test on either input graph gives a quick positive answer when $G_1$ or $G_2$ is already $K_r-e$-free; neither test is needed by the closure algorithm. Starting from $G_1$, it computes the least feasible extension, possibly $G_1$ itself. This distinction matters when $G_2=K_n$: membership of the complete upper graph is automatic, but the least extension is the information needed for the probe application. The standard restriction to the connected components of $G_1$ reduces the relevant upper graph (Proposition~\ref{prop:components}); for $r>4$ it can also reduce the number of potential clique-indexed worlds.

For additional examples and a more extensive exposition in Spanish, see Vekselman's undergraduate thesis draft~\cite{vekselmanThesis2026}. The present paper is self-contained and includes all proofs required for its results.

\section{The local characterization and forced-edge closure}\label{sec:local}
For a graph $J$ and a clique $S$, put $N_J(S)=\bigcap_{s\in S}N_J(s)$. A \emph{cluster graph} is a disjoint union of cliques, equivalently an induced-$P_3$-free graph. Let $k=r-3$ throughout. We write $\cl_t(J)$ for the set of $t$-cliques of $J$ and $c_t(J)=|\cl_t(J)|$.

\begin{lemma}[Local characterization]\label{lem:local}
A graph $J$ is induced-$K_r-e$-free if and only if $J[N_J(S)]$ is a cluster graph for every $k$-clique $S$ of $J$.
\end{lemma}
\begin{proof}
If $J[N_J(S)]$ contains an induced path $a-b-c$, then $S\cup\{a,b,c\}$ induces $K_r-e$, with $ac$ as its missing edge. Conversely, if $J$ has an induced $K_r-e$ with nonadjacent vertices $a,c$, choose $k=r-3$ vertices among the remaining $r-2$ vertices as $S$, and denote the unique vertex left over by $b$. The vertices $a$, $b$, and $c$ are all adjacent to every vertex of $S$; moreover, $b$ is adjacent to both $a$ and $c$. Hence $a-b-c$ is an induced $P_3$ in $J[N_J(S)]$.
\end{proof}

For $r=4$, this is the familiar characterization of diamond-free graphs by neighborhoods that induce disjoint unions of cliques; see, for example, Grandoni~\cite{grandoni2004}. The lemma states the corresponding characterization for every $r\ge4$.

We record the following standard component restriction to make explicit which upper graph enters our complexity bounds.
\begin{proposition}[Restriction to input components]\label{prop:components}
Let $C_1,\ldots,C_q$ be the connected components of $G_1$, and let $G_2'$ have vertex set $V$ and edge set
\[
E_2'=\{uv\in E_2: u,v\in C_i\text{ for some }i\}.
\]
Then $(G_1,G_2)$ is feasible if and only if $(G_1,G_2')$ is feasible. When feasible, the two instances have the same least solution, which contains no edge between distinct $C_i$.
\end{proposition}
\begin{proof}
Every sandwich for $(G_1,G_2')$ is also a sandwich for $(G_1,G_2)$. Conversely, given a feasible sandwich $H$ for $(G_1,G_2)$, retain its edges within each $C_i$ and delete all edges between different $C_i$. Each resulting subgraph on $C_i$ is induced in $H$ and hence $K_r-e$-free. Their disjoint union remains $K_r-e$-free because $K_r-e$ is connected. It lies between $G_1$ and $G_2'$, proving the reverse implication. The least solution for $(G_1,G_2')$ is contained in every feasible solution of $(G_1,G_2)$: restrict such a solution first, then apply leastness in the restricted instance (the least-solution property is established below in Theorem~\ref{thm:correct}). It is itself feasible for the original instance, so the two least solutions coincide.
\end{proof}

The same standard restriction works for any target class that is both hereditary and closed under disjoint union, such as chordal graphs and interval graphs. Computing the components and filtering an explicitly listed $E_2$ takes $O(n+m_2)$ time. If $m_2'=|E_2'|$ and $\alpha'=\alpha(G_2')$, applying Theorem~\ref{thm:main} after filtering gives $O(n+m_2+(\alpha')^{r-3}m_2')$ time. This bound retains the cost of reading the original edge list. For an implicitly specified complete upper graph, the restricted graph is simply a disjoint union of complete graphs on the $C_i$.

The characterization extends directly to a forcing rule. In any $K_r-e$-free completion containing a current graph $H$, if two vertices belong to the same connected component of $H[N_H(S)]$, they must be adjacent. Thus any missing edge inside such a component is mandatory. If a mandatory edge lies outside $E_2$, the instance is infeasible. Starting with $H=G_1$ and repeatedly adding mandatory edges defines a monotone closure. The result, when feasible, is the least feasible sandwich by edge inclusion: by induction every added edge belongs to every feasible completion; at a fixed point Lemma~\ref{lem:local} guarantees feasibility. This observation is the logical core of the closure in \cite{dantas2011}. What remains is to execute it without examining all nonedges or rebuilding common neighborhoods.

We use four conceptual edge states: \emph{forbidden} for $E(K_n)\setminus E_2$, \emph{raw} for $E_2\setminus E_1$ not yet forced, \emph{pending} for forced edges awaiting processing (including $E_1$ initially), and \emph{done} for processed edges. The current graph $H$ consists of pending and done edges. A pending edge is already an edge of $H$, although its effect on the local data structures has yet to be propagated. No edge leaves $H$; an edge is enqueued at most once.

For each potential base $S\in\cl_k(G_2)$, we call the local data structure associated with $S$ its \emph{world}. Its fixed universe is $U_S=N_{G_2}(S)$. A world becomes active when $S$ is a clique of done edges; it then maintains components on the \emph{active vertices}, those joined by done edges to every member of $S$. These are the local counterparts of the components of $H[N_H(S)]$ used by the closure, with pending edges reflected only after processing. Potential vertices can have records in $U_S$ before they become active; keeping them separate from active vertices is what allows the frontier and component structures to be allocated sparsely over $U_S$ rather than over all $n$ vertices.

\section{Filtered frontiers and component merging}\label{sec:frontier}
Fix a potential base $S\in\cl_k(G_2)$. Initially reserve one singleton record for each $x\in U_S$. As events arrive, some records become active. For every active component $C$ of the graph on active vertices, maintain a filtered frontier $F_S(C)$ of keys $z\in U_S\setminus C$. A key $z$ occurs precisely when $z$ is adjacent in $G_2$ to \emph{every} vertex of $C$; its associated linked list $L_C(z)$ holds exactly the distinct edges $xz$ for $x\in C$. Each singleton's frontier is created from the $(k+2)$-cliques $S\cup\{x,z\}$ of $G_2$. A potential singleton can be initialized before its activation.

\begin{lemma}[Exact frontier invariant]\label{lem:frontier}
For each active component $C$ and each $z\in U_S\setminus C$,
\[
z\in F_S(C)\quad\Longleftrightarrow\quad xz\in E_2\ \text{for every }x\in C.
\]
If $z\in F_S(C)$, the list $L_C(z)$ consists of exactly the $|C|$ edges $xz$, one for each $x\in C$.
\end{lemma}
\begin{proof}
For $C=\{x\}$ this is the definition of a singleton frontier. Distinct components $C_1,C_2$ of the same world are disjoint. On merging them, retain exactly the keys in $F_S(C_1)\cap F_S(C_2)$ and splice their lists. An external vertex is adjacent to all of $C_1\cup C_2$ exactly when it is a key on both sides, and the spliced lists contain one edge for each vertex of their disjoint union. No vertex newly inside the union belongs to both lists of keys. Induction completes the proof.
\end{proof}

To merge two components, let $C_1$ have no more frontier keys than $C_2$ and retain $C_1$ as the representative. The procedure has two phases.
\begin{enumerate}
\item \emph{Test and force.} Scan the keys $z$ of $F_S(C_1)$, marking each in a scratch array with a pointer to its entry. Whenever the component map assigns $z$ the identifier of $C_2$, add $|L_{C_1}(z)|$ to a counter and force the raw edges in that list. No vertex of $C_2$ is enumerated directly. If the counter is not $|C_1||C_2|$, report infeasibility.
\item \emph{Update.} Scan $F_S(C_2)$ and, for each key also marked in $F_S(C_1)$, splice the two lists and retain that key. Discard unmatched entries on both sides, so the new frontier has exactly the intersection of the two key sets. Clear the scratch locations, relabel the vertices of $C_2$ with the surviving identifier by traversing its member list, append that list to the member list of $C_1$ in $O(1)$ time (member lists keep head and tail pointers), and update the component size.
\end{enumerate}
Forcing raw edges before the count is checked is harmless because failure ends the algorithm; the writes may instead be deferred. The update intersects the frontier key sets as required by Lemma~\ref{lem:frontier}. Frontiers are iterable keyed records, and the scratch array indexes local vertex positions. Edge states reside in the records of the edges of $G_2$ and are accessed through the frontier lists; the merge counter detects a missing crossing edge without looking up nonedges. No $n\times n$ adjacency matrix or dense array per world is required.

\begin{lemma}[Merge test]\label{lem:test}
Every active component is a clique in $G_2$. The counter in the preceding merge equals $|C_1||C_2|$ if and only if every crossing edge between $C_1$ and $C_2$ lies in $E_2$, equivalently if and only if $C_1\cup C_2$ is a clique in $G_2$. If equality holds, all $|C_1||C_2|$ crossing edges are available in the scanned lists.
\end{lemma}
\begin{proof}
By Lemma~\ref{lem:frontier}, a key $z\in C_2$ occurs in $F_S(C_1)$ exactly when $z$ is adjacent in $G_2$ to all of $C_1$. Each such key contributes exactly $|C_1|$. Equality holds precisely when every vertex of $C_2$ occurs; the lists then supply every crossing edge. Singletons are cliques, and a merge completes only when this equality holds, so by induction every active component is a clique in $G_2$ and the two formulations of the test coincide.
\end{proof}

For the cost analysis let $f_S(x)=|F_S(\{x\})|$, measured at initial construction, and put $B_S=\sum_{x\in U_S}f_S(x)$. A surviving frontier only loses keys (never gains new ones), so if $x$ represents a component $C$, then $|F_S(C)|\le f_S(x)$. Every successful merge scans $O(|F_S(C_1)|+|F_S(C_2)|)$ keys; since the smaller frontier survives, this is $O(|F_S(C_2)|)$. Charge that work to the representative of $C_2$, who is discarded permanently. The same charge pays for relabeling: by Lemma~\ref{lem:test}, every vertex of $C_2$ was a key of $F_S(C_1)$, whence
\[
|C_2|\le|F_S(C_1)|\le|F_S(C_2)|.
\]
Consequently, all key scans and relabelings in this world cost $O(B_S)$ in total. The frontier lists contain $B_S$ edge cells at construction. Splicing moves cells between lists without copying them. When a cell is examined as a crossing edge in a successful merge, its external endpoint becomes internal to the merged component. Its key is omitted from the new frontier, so that cell cannot appear in any later frontier or be examined again. A cell removed with a discarded key is likewise discarded only once. Hence scanning or deleting cells in all successful merges costs $O(B_S)$. A single failing merge may scan surviving cells once more, but there is at most one failing merge in the entire algorithm and its cost is at most $O(B_S)$ for its world. Already-joined endpoints cause an $O(1)$ return and are charged to their unique triggering clique incidence.

\begin{lemma}[Total merge cost]\label{lem:mergecost}
Across all worlds, component merges, including a possible final failed merge, cost $O(c_{k+2}(G_2))$ time and require $O(c_{k+2}(G_2))$ frontier storage.
\end{lemma}
\begin{proof}
The preceding charging argument gives $O(B_S)$ per world. Each $(k+2)$-clique $Q$ contributes two directed singleton cells for each choice of $k$-element base $S\subset Q$. Thus
\[
\sum_{S\in\cl_k(G_2)}B_S
=2\binom{k+2}{2}c_{k+2}(G_2)
=(k+1)(k+2)c_{k+2}(G_2).
\]
The multiplier is constant for fixed $r$. The one unsuccessful merge is covered by its world's initial cells.
\end{proof}

\section{Clique events and the sandwich algorithm}\label{sec:events}
We now specify the events that make a world active and trigger its merges. Enumerate all cliques of $G_2$ of orders $1$ through $r-1=k+2$. For each clique $Q$ of order at least three, record its $|Q|$ immediate subcliques $Q\setminus\{x\}$ and the reverse incidence lists. Vertices start \emph{complete}; an edge becomes complete when it becomes done. A higher-order clique becomes complete upon receiving completion notifications from all of its immediate subcliques. Each node receives at most one notification per predecessor and fires once. After an edge has become done and its entire event cascade has returned, a node is complete exactly when all its edges are done: this follows by induction on its size, since every pair belongs to an immediate subclique of each clique of order at least three.

On completion of a $k$-clique $S$, activate its world (for $k=1$, vertices are complete initially and every vertex world is activated at initialization). On completion of a $(k+1)$-clique $S\cup\{x\}$, activate $x$ as a singleton in the world of each $k$-subclique $S$. On completion of a $(k+2)$-clique $Q=S\cup\{x,y\}$, call the merge procedure on $x,y$ in the world of each $k$-subclique $S$ of $Q$. An event performs its local action \emph{before} notifying successor cliques. Thus the world of $S$ is active before its first vertex event, and both $x$ and $y$ are active before their merge event; the edge $xy$ is done by that event. When an edge is popped from the pending queue, mark it done and fire its completion event: first its own local action (activating each endpoint in the world of the other if $k=1$, activating the world indexed by the edge if $k=2$, and performing no immediate local action if $k>2$), then notifications through the clique incidence lists. Newly forced raw edges enter the pending queue exactly once. The event ordering also applies when several cliques complete in the same cascade.

\begin{remark}[The diamond case]
For $r=4$, $k=1$: worlds are indexed by vertices, $(k+1)$-clique events are edges and activate neighbors, and $(k+2)$-clique events are triangles and trigger component merges.
\end{remark}

We call the state immediately before removing a pending edge, or immediately after the completion cascade for a newly done edge returns, a \emph{synchronization point}. Events may still be outstanding inside a cascade, so the full correspondence between worlds and the graph of done edges is asserted only at these points.

\begin{lemma}[World invariant]\label{lem:world}
At each synchronization point, each active world $S$ has exactly the vertices $A_S=N_D(S)$, where $D$ is the graph of done edges. Its maintained components are exactly the connected components of $D[N_D(S)]$.
\end{lemma}
\begin{proof}
The completion event for $S\cup\{x\}$ occurs precisely after every edge from $x$ to $S$ is done, and inserts $x$ once. For active vertices $x,y$, their local edge is in $D$ exactly when $S\cup\{x,y\}$ is a clique of done edges. In that case its $(k+2)$-clique event has merged their maintained components. Conversely, each merge event certifies a done edge $xy$ between active vertices. These statements imply both the vertex equality and equality of the partitions into connected components. Furthermore, during an unfinished cascade, each maintained component is contained in a connected component of the current done-edge graph on the vertices activated so far: every merger processed thus far was triggered by a done local edge.
\end{proof}

\begin{theorem}[Correctness and least solution]\label{thm:correct}
The algorithm reports infeasibility exactly when no $K_r-e$-free sandwich exists. On a feasible input it returns the unique least feasible sandwich $H_{\min}$ under edge inclusion; different processing orders return the same graph.
\end{theorem}
\begin{proof}
Fix any feasible sandwich $H^*$ and inductively assume every pending or done edge belongs to $H^*$. At a merge event, the base $S$ is a clique of done edges and its active vertices belong to $N_{H^*}(S)$. Each already maintained component is connected through done edges, even during a cascade. The merge joins components across a done edge $xy$, so their union lies within one component of $H^*[N_{H^*}(S)]$. By Lemma~\ref{lem:local}, that component is a clique. Hence every crossing edge is in $H^*$ and thus in $E_2$; the merge test succeeds and each newly forced edge belongs to $H^*$. This establishes the induction, including during cascades. A failed merge therefore proves infeasibility. The algorithm terminates: each edge enters the pending queue at most once and becomes done at most once, and each clique node fires at most once.

If the algorithm terminates successfully, all pending edges have become done. Every base clique of its final graph has a complete node and an active world. By Lemma~\ref{lem:world}, its maintained components are those of the induced common neighborhood. Initially each maintained component is a singleton. Whenever two components merge, the successful test forces every edge between them; when the queue becomes empty, all these edges are done. By induction over the successful merges, every maintained component therefore induces a clique of done edges, including pairs brought together by earlier merges. Lemma~\ref{lem:local} implies the final graph is $K_r-e$-free; it contains $E_1$ and is a subgraph of $G_2$. The induction above proves its edges occur in every feasible $H^*$, establishing leastness. Applying this containment to the outputs of two successful processing orders in both directions proves equality.
\end{proof}

\section{Construction and complexity}\label{sec:complexity}
Clique nodes must have unique identities: the same clique can be encountered through different immediate predecessors. Represent each $t$-clique by an increasing $t$-tuple of vertex identifiers. For fixed $r$, radix-sort these tuples at each level, assign consecutive identifiers, and produce the $t$ immediate-subclique lookup requests of each $(t+1)$-clique. Radix-sort the requests and merge them against the sorted canonical $t$-clique tuples to resolve every forward and backward pointer. Counting sort of each coordinate takes $O(n+c_t(G_2)+c_{t+1}(G_2))$ time; there are only $O(r)$ levels and coordinates. At level $k$, the requests from $(k+1)$-cliques associated with the same base $S$ are consecutive. Number them locally to give every pair $(S,x)$ its position in $U_S$. For every $(k+2)$-clique $Q$ and each of its constant number of decompositions $Q=S\cup\{x,y\}$, form lookup requests for $(S,x)$ and $(S,y)$; sort these requests with the local-position records and store the two resolved positions in the incidence record of $Q$. When $Q$ fires its merge event, those positions directly index the component entries for $x$ and $y$ in the world of $S$ in $O(1)$ time. Sparse arrays of length $|U_S|$, not length $n$, serve as component maps and merge scratch space. The sum of these lengths is $(k+1)c_{k+1}(G_2)$; for $k=1$ it is $2m_2$. The same decompositions build the singleton frontiers: $Q=S\cup\{x,z\}$ inserts a cell for key $z$ into $F_S(\{x\})$ and a symmetric cell for key $x$ into $F_S(\{z\})$. Both cells hold a direct reference to the canonical record of the edge $xz$ of $G_2$, obtained by sorting the pair requests $\{x,z\}$ and merging them against the sorted canonical $2$-clique tuples in $O(n+m_2+c_{k+2}(G_2))$ time. Edge-state queries and updates through a frontier cell therefore take $O(1)$ time without hashing or an adjacency matrix. Initialization of the global vertex and edge tables takes $O(n+m_2)$. We work in the standard word-RAM model with word size $w=\Theta(\log n)$ bits, where the hidden constant may depend on the fixed value of $r$, and with vertex identifiers in $\{0,\ldots,n-1\}$. Indeed, the algorithm uses $M=O(n^{r-1})$ memory cells, and hence $\log M=O(\log n)$; the word size can therefore be chosen large enough to address all of them. Thus all clique identifiers, array indices, and pointers fit in one word; arithmetic, indexing, and pointer operations take $O(1)$ time, as required by the counting-sort construction above.

After clique enumeration, each node and immediate incidence is processed $O(1)$ times (for fixed $r$); each frontier cell and all merge operations together cost $O(c_{k+2}(G_2))$ by Lemma~\ref{lem:mergecost}. Thus the \emph{post-enumeration} time and storage are
\begin{equation}\label{eq:output}
O\left(n+\sum_{t=2}^{r-1}c_t(G_2)\right).
\end{equation}
The clique-listing algorithm of Chiba and Nishizeki \cite{chiba1985} enumerates $t$-cliques, for each fixed $t\ge3$, in $O(n+\alpha(G_2)^{t-2}m_2)$ time, and the number of such cliques obeys the same bound without the $n$ term. List edges directly for $t=2$. Since $r$ is fixed and $\alpha\ge1$ when $m_2>0$, enumerating all required levels and applying~\eqref{eq:output} proves Theorem~\ref{thm:main}. The $O(n)$ term covers isolated vertices and empty graphs. Explicitly storing the clique incidences and frontier cells yields the stated space bound. For $r=4$, $c_3(G_2)=O(\alpha m_2)$, and each triangle contributes exactly six initial frontier cells.

For every fixed $r\ge4$, bounded arboricity of $G_2$ gives $O(n+m_2)$ time. The next observation explains when sparsity makes even the closure unnecessary.

Put $a=\alpha(G_2)\ge2$. Any $s$-vertex subgraph of $G_2$ has at most $a(s-1)$ edges, so $\omega(G_2)\le2a$. The edge count of the obstruction sharpens this clique bound: $K_{2a+1}-e$ has $a(2a)+a-1>a(2a)$ edges and therefore cannot occur even as a subgraph of $G_2$. For every $r\ge2a+1$, $K_r-e$ contains $K_{2a+1}-e$ on a subset of its vertices; thus neither $G_2$ nor $G_1\subseteq G_2$ can contain an induced $K_r-e$. Consequently $G_1$ is the least solution. This threshold is tight: the instance $(G_1,G_2)=(K_{2a}-e,K_{2a})$ with $r=2a$ requires completing the missing edge, and $\alpha(K_{2a})=a$. If $a=1$, the same triviality holds for every $r\ge4$ because forests contain no triangle. Planarity yields a sharper conclusion than $a\le3$ alone: for $r\ge6$ the instance is always solved by $G_1$, while for $r=5$ an induced $K_5-e$ in $G_1$ would require its missing edge to belong to $E_2$, so $G_2$ would contain $K_5$, contradicting planarity. Therefore in that case feasibility is equivalent to $G_1$ being $K_5-e$-free. For $r=4$, a planar upper graph may still permit nontrivial completion.

\section{Direct static recognition}\label{sec:static}
The fixed-input version needs no pending queue, incomplete world or frontier. Given a single graph $G$, enumerate its cliques of sizes $k,k+1,k+2$. For each $(k+1)$-clique $Q$ and each $k$-subclique $S\subset Q$, writing $Q=S\cup\{x\}$, register the local vertex $x$ in the world $S$; for each $(k+2)$-clique $Q$ and each of its decompositions $Q=S\cup\{x,y\}$, register the local edge $xy$ in the world $S$. Every resulting local graph is precisely $G[N_G(S)]$. Find its connected components by depth-first search and check that every component $C$ has exactly $\binom{|C|}{2}$ edges. Lemma~\ref{lem:local} proves this recognizes induced-$K_r-e$-free graphs. Radix sorting the fixed-size base tuples groups all local incidences, as above. The total time, and the space when incidences are stored, are $O(n+\alpha(G)^{r-3}|E(G)|)$. For $r=4$ this reproduces a known arboricity-sensitive recognition bound \cite{lin2012}. This static construction is included as a direct consequence of the characterization, not as a separate claim of algorithmic priority for general induced-pattern detection.

\section{Probe \texorpdfstring{$K_r-e$}{Kr-e}-free graphs}\label{sec:probe}
A probe partition of $G=(V,E)$ comprises a set $P$ of probes and an independent set $N$ of nonprobes; a valid completion adds edges only within $N$ and is $K_r-e$-free. Probe diamond-free graphs, the $r=4$ case, were studied by Bonomo et al.~\cite{bonomo2015}. A partition may be prescribed or may have to be found. Run our sandwich algorithm on $(G,K_n)$, which is always feasible since $K_n$ is $K_r-e$-free. Let $H_{\min}$ denote its least output, $F=E(H_{\min})\setminus E$, and $N_F$ be the set of endpoints of the edges in $F$.

\begin{theorem}[Probe consequence]\label{thm:probe}
For every fixed $r\ge4$: (i) $G$ admits a nonpartitioned probe $K_r-e$-free completion if and only if $N_F$ is independent in $G$; (ii) a prescribed partition $(P,N)$ admits such a completion if and only if $N$ is independent in $G$ and $N_F\subseteq N$. In either positive case, $F$ is the unique least completion under edge inclusion, and $N_F$ is contained in every feasible nonprobe set. Both variants can be decided in $O(n^{r-1})$ time and space.
\end{theorem}
\begin{proof}
Any valid completion $H^*$ is a feasible sandwich between $G$ and $K_n$, so Theorem~\ref{thm:correct} gives $F\subseteq E(H^*)\setminus E$. The latter edges have both endpoints in the independent nonprobe set of its probe partition, showing necessity in both cases. Conversely, if $N_F$ is independent, the partition $(V\setminus N_F,N_F)$ together with the edges $F$ witnesses the nonpartitioned case. If a prescribed independent $N$ contains $N_F$, the same completion respects that partition. The containment of $F$ in every completion gives the leastness and containment claims. Independence and containment can be tested in $O(n+|E|)$ time. For $G_2=K_n$, $m_2=\Theta(n^2)$ and $\alpha(K_n)=O(n)$, giving $O(n^{r-1})$ by Theorem~\ref{thm:main}.
\end{proof}

Algorithm~1 of Couto et al.~\cite{couto2018} maintains a graph $G^*$, initially $G$, and a set $N$, initially empty. Whenever $G^*$ contains an induced $K_r-e$ with missing edge $xy$, it checks in the original graph $G$ that neither endpoint has a neighbor already in $N$. If this test fails, it rejects; otherwise, it adds $x,y$ to $N$ and $xy$ to $G^*$. Thus, their algorithm applies the same forced-edge closure principle while checking incrementally that the endpoints of the inserted edges can form an independent nonprobe set.

The direct $O(n^r m_2)$ sandwich bound of Dantas et al.~\cite{dantas2011} becomes $O(n^{r+2})$ for $G_2=K_n$. Couto et al.~\cite{couto2018} state $O(n^r|E(G)|)$ for the algorithm just described, where $|E(G)|$ counts input edges, but the paper does not provide an implementation-level analysis establishing this bound. In particular, it does not specify how the next obstruction is located after an insertion, and we are not aware of an analysis that supplies this missing implementation detail. A natural direct implementation that exhaustively searches for an induced $K_r-e$ in $O(n^r)$ time before each of at most $O(n^2)$ insertions takes $O(n^{r+2})$ time.

\paragraph{Practical probe implementation.} The final independence test is convenient for stating and proving Theorem~\ref{thm:probe}. In the algorithm of Couto et al. just described, independence is instead checked incrementally, but both endpoints are tested after every insertion, including endpoints already in $N$. We can avoid these repeated tests by integrating the check into our sparse sandwich algorithm on input $(G,K_n)$ at its \emph{raw-to-pending} transition: immediately when a new edge $uv\notin E(G)$ is forced, process each endpoint only if it is unmarked. Mark that endpoint and scan its adjacency list in the original $G$ for a marked neighbor; if one exists, return \emph{no} and stop the entire algorithm before processing further closure events. Already marked endpoints are skipped, so each adjacency list is scanned at most once. The original edges of $G$, which enter the queue initially, never trigger this test. Every newly forced edge belongs to any feasible completion, so a conflict rules out every independent nonprobe set. If the closure finishes without a conflict, the marked set is $N_F$ and is independent: return \emph{yes}, with the resulting completion, without assembling $N_F$ afterward or performing a separate final test. The checks cost $O(n+|E(G)|)$ overall; early rejection can save closure work on negative instances without changing the worst-case bound.

When the upper graph $K_n$ is given implicitly, Proposition~\ref{prop:components} permits processing only the complete graphs on the connected components $C_i$ of $G$. Writing $n_i=|C_i|$, this yields the refined bound $O(n+\sum_i n_i^{r-1})$ for the unrestricted probe closure, including construction of these local complete upper graphs. The worst-case bound of Theorem~\ref{thm:probe} is recovered when $G$ is connected.

For a prescribed partition, test first that $N$ is independent and set $G_2=(V,E\cup\binom N2)$. A feasible $K_r-e$-free sandwich between $G$ and this $G_2$ is exactly a completion respecting the partition. This direct reduction takes $O(n+\alpha(G_2)^{r-3}|E(G_2)|)$ time and space by Theorem~\ref{thm:main}; for $r=4$ the bound is $O(n+\alpha(G_2)|E(G_2)|)$. It can be smaller than the worst-case bound obtained with $K_n$. The single unrestricted closure used in Theorem~\ref{thm:probe}, on the other hand, exposes the canonical minimum $N_F$ for both versions.

For $r=4$, the unrestricted closure's worst-case $O(n^3)$ probe bound is superseded by the specialized $O(n+nm)$ nonpartitioned diamond-free algorithm of Grippo and Lin \cite{grippoLin2026}. Their algorithm computes a nonprobe set $N_{\min}$ contained in the nonprobe set of every valid partition, together with a completion adding edges only inside $N_{\min}$. Consequently, a prescribed independent set $N$ admits a completion exactly when $N_{\min}\subseteq N$: necessity follows from containment, and sufficiency follows because the computed completion also uses only vertices of $N$.

\section{Concluding remarks}
The bottleneck in the classical forced-edge closure is locating the edges that must be inserted after two local components meet. A filtered frontier records precisely those external vertices still adjacent to every member of a component; intersecting the frontier key sets maintains this condition, while charging discarded components and initial edge cells bounds all work by the number of $(r-1)$-cliques of the upper graph. Clique completion events extend the construction from vertex-indexed neighborhoods for diamonds to neighborhoods indexed by $(r-3)$-cliques, without scanning inactive worlds. Applied to probe recognition for fixed $r$, this implementation gives $O(n^{r-1})$ time for both partitioned and nonpartitioned inputs, improving the $O(n^{r+2})$ bound obtained from the direct sandwich closure. The same filtered frontier and clique events account for the improvement in the sandwich and probe settings.

\end{document}